\documentclass[11pt]{article}

\usepackage{amsfonts,latexsym,amsthm,amssymb,amsmath,amscd,euscript,tikz,mathtools}
\usepackage{framed}
\usepackage[margin=1in]{geometry}
\usepackage{color}
\usepackage[colorlinks=true,citecolor=blue,linkcolor=blue]{hyperref}
\usepackage{enumitem}
\usepackage{siunitx}
\usepackage{textcomp}
\usepackage{physics}
\usepackage{mathrsfs}
\usepackage{dsfont}
\usepackage[ruled,vlined,linesnumbered]{algorithm2e}
\usepackage{titlesec}
\usepackage{tikz-cd}
\usepackage{thmtools}
\usepackage{thm-restate}
\usepackage{cleveref}
\usepackage{stmaryrd}
\usepackage{graphicx}
\usetikzlibrary{positioning,chains,fit,shapes,calc}

\allowdisplaybreaks[1]

\newtheorem{theorem}{Theorem}[section]
\newtheorem{proposition}[theorem]{Proposition}
\newtheorem{lemma}[theorem]{Lemma}

\newtheorem{corollary}[theorem]{Corollary}

\newtheorem{definition}[theorem]{Definition}

\newtheorem{remark}[theorem]{Remark}

\theoremstyle{remark}

\newcommand{\cC}{\mathcal{C}}

\newcommand{\cN}{\mathcal{N}}
\newcommand{\cP}{\mathcal{P}}

\newcommand{\bC}{\mathbb{C}}
\newcommand{\bF}{\mathbb{F}}

\newcommand{\bN}{\mathbb{N}}

\newcommand{\bR}{\mathbb{R}}

\newcommand{\1}{\mathds{1}}

\newcommand{\Enc}{\operatorname{Enc}}
\newcommand{\Dec}{\operatorname{Dec}}

\newcommand{\evl}{\operatorname{ev}}

\newcommand{\CSS}{\operatorname{CSS}}

\newcommand{\MF}[1]{\underline{#1}}
\newcommand{\LAT}[1]{\overline{#1}}

\newcommand{\nc}{\newcommand}

\nc{\on}{\operatorname}
\nc{\Spec}{\on{Spec}}
\nc{\Aut}{\textit{Aut}}
\nc{\id}{\textit{id}}
\nc{\chr}{\on{char}}
\nc{\im}{\on{im}}
\nc{\Hom}{\on{Hom}}
\nc{\lcm}{\on{lcm}}
\nc{\dual}[1]{\prescript{t}{}{#1}}
\nc{\transpose}[1]{{#1}^{\intercal}}
\nc{\Sym}{\on{Sym}}
\nc{\End}{\on{End}}
\nc{\stab}{\on{stab}}
\nc{\Li}{\on{Li}}
\nc{\spn}{\on{span}}
\nc{\sgn}{\on{sgn}}
\nc{\supp}{\on{supp}}
\nc{\Unif}{\on{Unif}}

\makeatletter
\newcommand\footnoteref[1]{\protected@xdef\@thefnmark{\ref{#1}}\@footnotemark}
\makeatother

\title{
  Asymptotically Good Quantum Codes \\ with Transversal Non-Clifford Gates\thanks{Research supported in part by a ONR grant N00014-24-1-2491 and a UC Noyce initiative award. L.~Golowich is supported by a National Science Foundation Graduate Research Fellowship under Grant No.~DGE 2146752.}
}
\author{Louis Golowich \\
  UC Berkeley \\
  \href{mailto:lgolowich@berkeley.edu}{\texttt{lgolowich@berkeley.edu}}
  \and
  Venkatesan Guruswami \\
  UC Berkeley \\
  \href{mailto:venkatg@berkeley.edu}{\texttt{venkatg@berkeley.edu}}
}
\date{August 2024}

\begin{document}


\maketitle
\thispagestyle{empty}

\begin{abstract}
  We construct quantum codes that support transversal $CCZ$ gates over qudits of arbitrary prime power dimension $q$ (including $q=2$) such that the code dimension and distance grow linearly in the block length. The only previously known construction with such linear dimension and distance required a growing alphabet size $q$ (Krishna \& Tillich, 2019). Our codes imply protocols for magic state distillation with overhead exponent $\gamma=\log(n/k)/\log(d)\rightarrow 0$ as the block length $n\rightarrow\infty$, where $k$ and $d$ denote the code dimension and distance respectively. It was previously an open question to obtain such a protocol with a contant alphabet size $q$.

  We construct our codes by combining two modular components, namely,
  \begin{enumerate}[label=(\roman*)]
  \item\label{it:ctoq} a transformation from classical codes satisfying certain properties to quantum codes supporting transversal $CCZ$ gates, and
  \item\label{it:concat} a concatenation scheme for reducing the alphabet size of codes supporting transversal $CCZ$ gates. For this scheme we introduce a quantum analogue of \textit{multiplication-friendly codes}, which provide a way to express multiplication over a field in terms of a subfield.
  \end{enumerate}
  We obtain our asymptotically good construction by instantiating \ref{it:ctoq} with algebraic-geometric codes, and applying a constant number of iterations of~\ref{it:concat}. We also give an alternative construction with nearly asymptotically good parameters ($k,d=n/2^{O(\log^*n)}$) by instantiating \ref{it:ctoq} with Reed-Solomon codes and then performing a superconstant number of iterations of \ref{it:concat}.
\end{abstract}





\section{Introduction}
A major challenge in fault-tolerant quantum computation is to efficiently perform non-Clifford gates. The Clifford gate set, generated by the Hadamard, phase, and $CNOT$ gates over qubits, can be efficiently simulated by a classical computer, and has various known fault-tolerant implementations. Universal quantum computation can then be obtained by adding a single non-Clifford gate such as $CCZ$ or $T$. However, such non-Clifford gates are also typically more difficult to implement fault-tolerantly.

A prominent tool for addressing this challenge is given by codes supporting \textit{transversal} non-Clifford gates, meaning that the logical action of the gate on the encoded information can be induced by applying the gate to disjoint sets of physical code qubits (see Section~\ref{sec:transgates} for the formal definition we use in this paper). For instance, codes supporting such transversal non-Clifford gates can be used to distill high-fidelity copies of certain ``magic'' states, which when combined with fault-tolerant Clifford circuits yield protocols for universal fault-tolerant quantum computation \cite{bravyi_universal_2005}.

Our main result, stated below, is a construction of quantum codes supporting transversal $CCZ$ gates with asymptotically good parameters, meaning that the code dimension and distance grow linearly with the block length. Recall here that for qudits of dimension some prime power $q$, the $CCZ$ gate is the 3-qudit gate given by\footnote{For non-prime $q$, we will need a slightly more general definition of gates $CCZ_q^a$ and $U_q^a$ for every $a\in\bF_q$; see Definition~\ref{def:gates}.}
\begin{equation*}
  CCZ_q = \sum_{j_1,j_2,j_3\in\bF_q}e^{2\pi i\tr_{\bF_q/\bF_p}(j_1j_2j_3)/p}\ket{j_1,j_2,j_3}\bra{j_1,j_2,j_3},
\end{equation*}
We will also make use of the single-qudit $U$ gate given by
\begin{equation*}
  U_q = \sum_{j\in\bF_q}e^{2\pi i\tr_{\bF_q/\bF_p}(j^3)/p}\ket{j}\bra{j}.
\end{equation*}
For primes $q\geq 5$, the $U_q$ gate yields universal quantum computation when combined with the Clifford gates \cite{cui_diagonal_2017}, and hence provides a $q$-ary analogue of the $T$ gate.

\begin{theorem}[Informal statement of Theorem~\ref{thm:constalph}]
  \label{thm:constalphinf}
  For every fixed prime power $q$, there exists an infinite family of quantum codes that support a transversal $CCZ$ gate over $q$-dimensional qudits, such that the code dimension $k=\Theta(n)$ and distance $d=\Theta(n)$ grow linearly in the block length $n$. Furthermore, if $q\geq 5$, these codes also support a transversal $U$ gate.
\end{theorem}

For magic state distillation, \cite{bravyi_magic-state_2012} showed that a quantum code of length $n$, dimension $k$, and distance $d$ over supporting a transversal $T$ gate over a binary alphabet yields a protocol for constructing magic $T$ states ($T\ket{+}$) with overhead $\log^\gamma(1/\epsilon)$ for $\gamma=\log(n/k)/\log(d)$, where $\epsilon>0$ is the target error rate. That is, $\kappa$ states with $<\epsilon$ infidelity to $T\ket{+}$ can be distilled from $\kappa\cdot\log^\gamma(1/\epsilon)$ states with some constant infidelity. These techniques extend to $U$ gates over larger $q$ (as observed by \cite{krishna_towards_2019}) and also to $CCZ$ gates (as observed by \cite{paetznick_universal_2013}). Hence Theorem~\ref{thm:constalphinf} yields $\gamma\rightarrow 0$ over arbitrary constant prime power alphabet sizes $q$.

Previously known constructions\footnote{\label{footnote:concurrent} See Remark~\ref{remark:concurrent} for relevant concurrent work.} achieving $\gamma\rightarrow 0$ for non-Clifford gates required alphabets growing with the block length \cite{krishna_towards_2019}. Meanwhile, previously known constructions with constant alphabet size $q$ only achieved $\gamma=\Omega(1)$ bounded away from $0$ (see e.g.~\cite{hastings_distillation_2018}).

To the best of our knowledge, Theorem~\ref{thm:constalphinf} provides the first known asymptotically good quantum codes supporting a transversal non-Clifford gate.\footnoteref{footnote:concurrent} However, \cite{krishna_towards_2019} constructed codes with linear dimension and distance supporting transversal $U_q$ gates over qudits of dimension $q=\Theta(n)$ growing linearly in the block length $n$. Their construction was based on Reed-Solomon codes, which consist of evaluations of polynomials and hence possess an algebraic structure amenable to transversal $CCZ$ and $U$ gates.

To prove Theorem~\ref{thm:constalphinf}, we generalize the construction of \cite{krishna_towards_2019} (see Theorem~\ref{thm:generaltrans}) to obtain quantum codes supporting transversal $CCZ,U$ gates from any classical codes $C$ for which $C$, its dual code $C^\perp$, and its triple-product code $C^{*3}$ are all asymptotically good. Here $C^{*3}$ denotes the code generated by the component-wise products of triples of codewords in $C$. We then instantiate the classical codes $C$ in this construction (in a black-box manner) using algebraic-geometric (AG) codes, which are similar to Reed-Solomon codes but give a constant instead of growing alphabet size $q'$. We choose $q'$ to be a power of the desired qudit dimension $q$, and then reduce the alphabet size from $q'$ to $q$ via concatenation.

For this alphabet reduction, we extend the theory of classical \textit{multiplication-friendly} codes to the quantum setting (see Section~\ref{sec:multfriend}). Such codes allow for concatenation while preserving the algebraic structure of a code, which in our setting means we can preserve support for a transversal $CCZ$ or $U$ gate. By constructing quantum multiplication-friendly codes with Reed-Solomon codes, we show how to perform such alphabet reduction with just a small loss in the code dimension and distance.

In fact, as described in Theorem~\ref{thm:RSconstalph} in Appendix~\ref{sec:RSconcat}, we are able to almost recover the parameters in Theorem~\ref{thm:constalphinf} using concatenated Reed-Solomon codes in place of AG codes. Specifically, using repeated concatenation, we obtain codes of dimension and distance
\begin{equation}
  \label{eq:RSconcatinf}
  k \geq \Omega\left(\frac{n}{2^{O(\log^*n)}}\right), \hspace{1em} d\geq\Omega\left(\frac{n}{2^{O(\log^*n)}}\right)
\end{equation}
supporting transversal $CCZ_q$ (and also $U_q$ for $q\geq 5$). Here $\log^*n$ denotes the extremely slow-growing function given by the number of times one must apply a logarithm to obtain a value $\leq 1$, starting from $n$. While these parameters are not quite optimal, the construction is much more elementary than in Theorem~\ref{thm:constalphinf}, as it avoids the need for AG codes, which are fairly involved to construct.

\begin{remark}
  \label{remark:concurrent}
  We recently became aware of an independent and concurrent work \cite{wills_constant-overhead_2024} that obtained a similar result as Theorem~\ref{thm:constalphinf}, also using AG codes. \cite{wills_constant-overhead_2024} apply these codes to perform magic state distillation with constant overhead. We note that our use of concatenation with multiplication-friendly codes is distinct from their approach, and may in fact help address some open questions they pose pertaining to alphabet reduction. For instance, \cite{wills_constant-overhead_2024} primarily consider codes over $\bF_{1024}$, and leave it as an open problem to obtain codes with transversal non-Clifford gates over arbitrarily small (e.g.~binary) alphabets.

  Independently, Christopher A.~Pattison and Quynh T.~Nguyen obtained a (still unpublished) result for low-overhead magic state distillation that uses ideas related to those in Theorem~\ref{thm:RSconstalph}.
\end{remark}

\subsection{Open Problems}
Our work raises some interesting open questions. Although our codes achieve asymptotically optimal dimension and distance while supporting transversal non-Clifford gates, they have high-weight stabilizers, that is, they are not LDPC. Constructing such codes that are LDPC could present new possibilities for fault-tolerance, as low-weight stabilizers are much easier to measure in a fault-tolerant manner.

Note that we give general methods for preserving transversal non-Clifford gates under concatenation, resulting in concatenated codes with near-optimal asymptotic parameters (e.g.~Theorem~\ref{thm:RSconstalph}). Concatenated codes, though not LDPC, do possess many low-weight stabilizers, and as such form the basis for many fault-tolerance schemes (e.g.~\cite{aharonov_fault-tolerant_1997,yamasaki_time-efficient_2024,pattison_hierarchical_2023}). It is an interesting question whether our results can yield improvements for such schemes, beyond simply through an application of improved magic state distillation.

Furthermore, while our paper focuses on achieving optimal asymptotic performance, it is also an interesting question to determine whether our results may yield improved codes or protocols for non-Clifford gates in finite-sized instances that may be relevant in near-term devices.

\subsection{Organization}
The main body of this paper proves Theorem~\ref{thm:constalphinf}, while the appendix provides the concatenation-based construction with parameters described by~(\ref{eq:RSconcatinf}).

Specifically, the organization is as follows. Section~\ref{sec:prelim} contains preliminary definitions and basic results. In Section~\ref{sec:latrans}, we generalize the construction of \cite{krishna_towards_2019}, and then instantiate it with AG codes to obtain asymptotically good codes supporting transversal $CCZ$ (and $U$) gates over large but constant-sized alphabets. In Section~\ref{sec:multfriend}, we describe classical multiplication-friendly codes and introduce a quantum analogue. We then show that classical multiplication-friendly codes yield quantum ones, though with poor distance. These classical constructions are nevertheless sufficient when we only need a constant-sized instance; Appendix~\ref{sec:qmultfriend} gives a more inherently quantum construction of  multiplication-friendly codes with better asymptotic performance. In Section~\ref{sec:alphred} we show that the concatenation of a code supporting transversal $CCZ_{q'}$ with a multiplication-friendly code yields a code supporting transversal $CCZ_q$ over a smaller alphabet size $q|q'$. Section~\ref{sec:agconcat} proves applies a constant number of iterations of such concatenation starting from an AG-based code to prove Theorem~\ref{thm:constalphinf}. Meanwhile, Appendix~\ref{sec:RSconcat} applies a superconstant number of iterations of such concatenation starting from a Reed-Solomon-based code to obtain codes satisfying the bounds~(\ref{eq:RSconcatinf}).

\section{Preliminaries}
\label{sec:prelim}
This section presents preliminary definitions and some basic results.

\subsection{Notation}
\label{sec:notation}
For $n\in\bN$, we let $[n]=\{1,2,\dots,n\}$. For vectors $x,y\in\bF_q^n$, we let $x\cdot y=\sum_{j\in[n]}x_iy_i$ denote the standard bilinear form, and we let $x*y=(x_iy_i)_{i\in[n]}\in\bF_q^n$ denote the component-wise product. For a set $A\subseteq\bF_q$, we let $A^c=\bF_q\setminus A$ denote the complement of $A$. For a finite set $S$, we let $\ket{S}=(1/\sqrt{S})\sum_{s\in S}\ket{s}$ denote the uniform superposition over the elements of $S$.

For functions $f:A\rightarrow B$ and $g:B\rightarrow C$, we let $g\circ f:A\rightarrow C$ be the standard composition $(g\circ f)(x)=g(f(x))$. Letting $\cP(S)$ denote the power set of a set $S$, then if instead we have functions $f:A\rightarrow\cP(B)$ and $g:B\rightarrow\cP(C)$, we let $g\circ f:A\rightarrow\cP(C)$ be the composition given by $(g\circ f)(x)=\bigcup_{y\in f(x)}g(y)$.

\subsection{CSS Codes}
This section presents basic definitions pertaining to codes. We begin with the definition of a classical code:

\begin{definition}
  A \textbf{classical (linear) code $C$} of \textbf{length $n$}, \textbf{dimension $k$}, and \textbf{alphabet size $q$} is a $k$-dimensional linear subspace $C\subseteq\bF_q^n$. The \textbf{distance $d$} of $C$ is defined as $d=\min_{z\in C\setminus\{0\}}|z|$. As a shorthand, we say that $C$ is an $[n,k]_q$ or an $[n,k,d]_q$ code. The \textbf{dual $C^\perp\subseteq\bF_q^n$} of $C$ is defined as $C^\perp=\{x\in\bF_q^n:x\cdot z=0\;\forall z\in C\}$.

  An \textbf{encoding function $\Enc$} for $C$ is a linear isomorphism $\Enc:\bF_q^k\xrightarrow{\sim}C$.
\end{definition}

Below we introduce some notation for bounded-degree polynomials, whose evaluations yield the well-known Reed-Solomon codes.

\begin{definition}
  \label{def:polynomials}
  Given a polynomial $f(X)\in\bF_q[X]$ and a set $A\subseteq\bF_q$, we define the vector $\evl_A(f)\in\bF_q^A$ by $\evl_A(f)=(f(x))_{x\in A}$. We let $\bF_q[X]^{<k}$ denote the set of polynomials of degree $<k$. We then extend $\evl_A$ (or $\evl_{A^c}$) to act on sets of polynomials. For instance, if $k\leq|A|$, then $\evl_A(\bF_q[X]^{<k})$ is a (classical) \textbf{punctured Reed-Solomon code} (or simply \textbf{Reed-Solomon code}) of alphabet size $q$, dimension $k$, and length $|A|$.
\end{definition}

It is well known that $\evl_A(\bF_q[X]^{<k})$ has parameters $[|A|,\; k,\; |A|-k+1]_q$, and that if $A=\bF_q$, then the dual is the Reed-Solomon code $\evl_{\bF_q}(\bF_q[X]^{<k})^\perp=\evl_{\bF_q}(\bF_q[X]^{<q-k})$.

In this paper, we restrict attention to quantum codes given by the well-known CSS paradigm, described below.

\begin{definition}
  A \textbf{quantum CSS code $Q$} of \textbf{length $n$}, \textbf{dimension $k$}, and \textbf{alphabet size} (or \textbf{local dimension}) $q$ is a pair $Q=\CSS(Q_X,Q_Z)$ of linear subspaces $Q_X,Q_Z\subseteq\bF_q^n$ such that $Q_X^\perp\subseteq Q_Z$ and $k=\dim(Q_Z)-\dim(Q_X^\perp)$. The \textbf{distance $d$} of $Q$ is defined as
  \begin{equation*}
    d = \min_{y\in(Q_X\setminus Q_Z^\perp)\cup(Q_Z\setminus Q_X^\perp)}|y|.
  \end{equation*}
  As a shorthand, we say that $C$ is an $[[n,k]]_q$ or an $[[n,k,d]]_q$ code.

  We sometimes refer to elements of $Q_X^\perp$ (resp.~$Q_Z^\perp$) as \textbf{$X$-stabilizers} (resp.~\textbf{$Z$-stabilizers}).

  A \textbf{$X$ (resp.~$Z$) encoding function $\Enc_X$ (resp.~$\Enc_Z$)} for $Q$ is an $\bF_q$-linear isomorphism $\Enc_X:\bF_q^k\xrightarrow{\sim}(Q_X/Q_Z^\perp)$ (resp.~$\Enc_Z:\bF_q^k\xrightarrow{\sim}(Q_Z/Q_X^\perp)$). Given a bilinear form $B:\bF_q^k\times\bF_q^k\rightarrow\bF_q$, we say that a pair $(\Enc_X,\Enc_Z)$ of $X,Z$ encoding functions are \textbf{compatible with respect to $B$} if it holds for every $x,z\in\bF_q^k$ and every $x'\in\Enc_X(x),z'\in\Enc_Z(z)$ that $B(x,z)=x'\cdot z'$.

  We write $Q=\CSS(Q_X,Q_Z;\Enc_Z)$ (resp.~$Q=\CSS(Q_X,Q_Z;\Enc_X,\Enc_Z)$) to refer to a CSS code with a specified $Z$ (resp.~$X$ and $Z$) encoding function.
\end{definition}

\begin{remark}
  We will sometimes view the domain of the encoding function $\Enc$ of a classical code or the encoding functions $\Enc_X,\Enc_Z$ of a quantum code of dimension $k$ as the space $\bF_{q^k}$. Such a view is valid because $\bF_{q^k}\cong\bF_q^k$, though this isomorphism is non-canonical, so we will explicitly note whenever we make this choice.
\end{remark}


\begin{lemma}
  \label{lem:findcompatible}
  Given an $[[n,k]]_q$ code $Q=\CSS(Q_X,Q_Z;\Enc_Z)$ and a nondegenerate bilinear form $B:\bF_q^k\times\bF_q^k\rightarrow\bF_q$, there exists an $X$ encoding function $\Enc_X$ that is compatible with $\Enc_Z$ with respect to $B$.
\end{lemma}
\begin{proof}
  Because $B$ is nondegenerate, we can choose bases $a_1,\dots,a_k$ and $b_1,\dots,b_k$ for $\bF_q^k$ with that $B(a_i,b_j)=\1_{i=j}$. Furthermore, because we have a nondegenerate bilinear form $\langle\cdot,\cdot\rangle:(Q_X/Q_Z^\perp)\times(Q_Z/Q_X^\perp)\rightarrow\bF_q$ given by $\langle x'+Q_Z^\perp,z'+Q_X^\perp\rangle=x'\cdot z'$, we can choose bases $e_1,\dots,e_k$ and $f_1,\dots,f_k$ for $Q_X/Q_Z^\perp$ and $Q_Z/Q_X^\perp$ respectively such that $\langle e_i,f_j\rangle=\1_{i=j}$.

  Define a matrix $M\in\bF_q^{k\times k}$ such that for $j\in[k]$, the $j$th column of $M$ is the representation of $\Enc_Z(b_j)$ in the basis $f_1,\dots,f_k$. Then for $i\in[k]$ define $\Enc_X(a_i)$ to be the $i$th row of the inverse matrix $M^{-1}\in\bF_q^{k\times k}$, where we interpret this row as the representation of an element of $Q_X/Q_Z^\perp$ in the basis $e_1,\dots,e_k$. We then extend $\Enc_X$ by linearity to a function $\Enc_X:Q_X/Q_Z^\perp\rightarrow\bF_q$.

  Because $M^{-1}M=I$ and $\langle e_i,f_j\rangle=\1_{i=j}$, it follows that for every $i,j\in[k]$, we have $\langle\Enc_X(a_i),\Enc_Z(b_j)\rangle=I_{i,j}=\1_{i=j}=B(a_i,b_j)$. Therefore by linearity, it holds for every $x,z\in\bF_q^k$ that $B(x,z)=\langle\Enc_X(x),\Enc_Z(z)\rangle$. Thus $\Enc_X,\Enc_Z$ are compatible with respect to $B$, as desired.
\end{proof}

We will make use of the following notion of concatenation of CSS codes.

\begin{definition}
  \label{def:concatenation}
  Consider CSS codes with specified $X,Z$ encoding functions
  \begin{align*}
    Q^{\mathrm{in}} &= \CSS(Q^{\mathrm{in}}_X,Q^{\mathrm{in}}_Z;\Enc^{\mathrm{in}}_X,\Enc^{\mathrm{in}}_Z) \hspace{1em} \text{ with parameters } [[n^{\mathrm{in}},k^{\mathrm{in}},d^{\mathrm{in}}]]_{q^{\mathrm{in}}} \\
    Q^{\mathrm{out}} &= \CSS(Q^{\mathrm{out}}_X,Q^{\mathrm{out}}_Z;\Enc^{\mathrm{out}}_X\Enc^{\mathrm{out}}_Z) \hspace{1em} \text{ with parameters } [[n^{\mathrm{out}},k^{\mathrm{out}},d^{\mathrm{out}}]]_{q^{\mathrm{out}}},
  \end{align*}
  such that $q^{\mathrm{out}}=(q^{\mathrm{in}})^{k^{\mathrm{in}}}$ so that $\bF_{q^{\mathrm{out}}}\cong\bF_{q^{\mathrm{in}}}^{k^{\mathrm{in}}}$, and also assume that
  \begin{equation*}
    (\Enc^{\mathrm{in}}_X:\bF_{q^{\mathrm{out}}}\rightarrow Q^{\mathrm{in}}_X/{Q^{\mathrm{in}}_Z}^\perp,\; \Enc^{\mathrm{in}}_Z:\bF_{q^{\mathrm{out}}}\rightarrow Q^{\mathrm{in}}_Z/{Q^{\mathrm{in}}_X}^\perp)
  \end{equation*}
  are compatible with respect to the the trace bilinear form $(x,z)\mapsto\tr_{\bF_{q^{\mathrm{out}}}/\bF_{q^{\mathrm{in}}}}(xz)$, and that
  \begin{equation*}
    (\Enc^{\mathrm{out}}_X:\bF_{q^{\mathrm{out}}}^{k^{\mathrm{out}}}\rightarrow Q^{\mathrm{out}}_X/{Q^{\mathrm{out}}_Z}^\perp,\; \Enc^{\mathrm{out}}_Z:\bF_{q^{\mathrm{out}}}^{k^{\mathrm{out}}}\rightarrow Q^{\mathrm{out}}_Z/{Q^{\mathrm{out}}_X}^\perp)
  \end{equation*}
  are compatible with respect to the standard bilinear form $(x,z)\mapsto x\cdot z$.

  Then the \textbf{concatenated code $Q=Q^{\mathrm{in}}\circ Q^{\mathrm{out}}$} is the $[[n^{\mathrm{in}}n^{\mathrm{out}},k^{\mathrm{in}}k^{\mathrm{out}},d^{\mathrm{in}}d^{\mathrm{out}}]]_{q^{\mathrm{in}}}$ CSS code specified by the $\bF_q^{\mathrm{in}}$-linear encoding functions
  \begin{align*}
    \Enc_X &= (\Enc^{\mathrm{in}}_X)^{\oplus n^{\mathrm{out}}}\circ\Enc^{\mathrm{out}}_X \\
    \Enc_Z &= (\Enc^{\mathrm{in}}_Z)^{\oplus n^{\mathrm{out}}}\circ\Enc^{\mathrm{out}}_Z,
  \end{align*}
  so that formally $Q=\CSS(\im(\Enc_X),\im(\Enc_Z);\Enc_X,\Enc_Z)$.
\end{definition}

As a point of notation above, we compose (encoding) functions mapping vectors to sets (or specifically cosets) in the natural way described in Section~\ref{sec:notation}, by simply outputting the union of the outputs of the second function applied to all outputs of the first.

A proof that the concatenated code in Definition~\ref{def:concatenation} is indeed a well-defined CSS code can for instance be found in \cite{bergamaschi_approaching_2022}.

We will also use the following notion of restriction of CSS codes, which simply adds some $Z$-stabilizers to a given code, thereby reducing the dimension without negatively affecting the distance.

\begin{definition}
  \label{def:restriction}
  Let $Q=\CSS(Q_X,Q_Z;\Enc_Z)$ be a $[[n,k,d]]_q$ CSS code. Given a subspace $S\subseteq\bF_q^k$, we define an $[[n,\dim(S),\geq d]]_q$ \textbf{restricted code $Q|_S=\CSS(Q_X,\im(\Enc_Z(S));\Enc_Z|_S)$}.
\end{definition}

\subsection{Transversal $CCZ$ and $T$-Like Gates}
\label{sec:transgates}
Below, we define the $CCZ_q$ gate over $q$-dimensional qudits, as well as a related gate we call $U_q$ that is similar to a $T$ gate, as described in \cite{krishna_towards_2019}.

\begin{definition}
  \label{def:gates}
  For a prime power $q=p^r$ and an element $a\in\bF_q$, define a unitary gate $CCZ_q:(\bC^{\bF_q})^{\otimes 3}\rightarrow(\bC^{\bF_q})^{\otimes 3}$ acting on three qubits of dimension $q$ by
  \begin{equation*}
    CCZ_q^a = \sum_{j_1,j_2,j_3\in\bF_q}e^{2\pi i\tr_{\bF_q/\bF_p}(a\cdot j_1j_2j_3)/p}\ket{j_1,j_2,j_3}\bra{j_1,j_2,j_3},
  \end{equation*}
  and define a unitary gate $U_q^a:\bC^{\bF_q}\rightarrow\bC^{\bF_q}$ acting on a single qudit of dimension $q$ by
  \begin{equation*}
    U_q^a = \sum_{j\in\bF_q}e^{2\pi i\tr_{\bF_q/\bF_p}(a\cdot j^3)/p}\ket{j}\bra{j}.
  \end{equation*}
  When $a=1$, we denote $CCZ_q^1=CCZ_q$ and $U_q^1=U_q$.
\end{definition}

We are interested in codes supporting transversal $CCZ_q$ and $U_q$ gates, meaning that by applying $CCZ_q$ or $U_q$ gates independently across the physical qudits in a code block (or for $CCZ_q$, across three code blocks), we induce a logical $CCZ_q$ or $U_q$ gate on all of the encoded qudits.

\begin{definition}
  \label{def:supporttrans}
  For $h\in[3]$, let $Q^{(h)}=\CSS(Q^{(h)}_X,Q^{(h)}_Z;\Enc^{(h)}_Z)$ be a $[[n,k]]_q$ CSS code. We say that the triple $(Q^{(1)},Q^{(2)},Q^{(3)})$ \textbf{supports a transversal $CCZ_q$ gate} if the following holds: there exists some $b\in\bF_q^n$ (that we call the \textbf{coefficients vector}) such that for every $z^{(1)},z^{(2)},z^{(3)}\in\bF_q^k$, it holds for every ${z^{(1)}}'\in\Enc_Z(z^{(1)}),\;{z^{(2)}}'\in\Enc_Z(z^{(2)}),\;{z^{(3)}}'\in\Enc_Z(z^{(3)})$ that
  \begin{equation}
    \label{eq:supporttransCCZ}
    \sum_{j\in[k]}z^{(1)}_jz^{(2)}_jz^{(3)}_j = \sum_{j\in[n]}b_j\cdot{z^{(1)}_j}'{z^{(2)}_j}'{z^{(3)}_j}'.
  \end{equation}
  If $Q^{(1)}=Q^{(2)}=Q^{(3)}$, we say that this code supports a transversal $CCZ_q$ gate.

  Similarly, we say that a single $[[n,k]]_q$ code $Q=\CSS(Q_X,Q_Z;\Enc_Z)$ \textbf{supports a transversal $U_q$ gate} with \textbf{coefficients vector $b$} if it holds for every $z\in\bF_q^k$ and every $z'\in\Enc_Z(z)$ that
  \begin{equation}
    \label{eq:supporttransU}
    \sum_{j\in[k]}z_j^3=\sum_{j\in[n]}b_j\cdot{z'_j}^3.
  \end{equation}
\end{definition}


The following basic claim shows that we may essentially restrict attention to constructing codes supporting a transversal $CCZ_q$ gate.

\begin{lemma}
  \label{lem:CCZtoU}
  If a code $Q$ supports a transversal $CCZ_q$ gate, then $Q$ supports a transversal $U_q$ gate.
\end{lemma}
\begin{proof}
  The claim follows directly from Definition~\ref{def:supporttrans}.
\end{proof}

The following basic lemma shows that Definition~\ref{def:supporttrans} is well-formulated.

\begin{lemma}
  Let $(Q^{(h)}=\CSS(Q^{(h)}_X,Q^{(h)}_Z;\Enc^{(h)}_Z))_{h\in[3]}$ be a triple of $[[n,k]]_q$ CSS codes that supports a transversal $CCZ_q$ gate with coefficients vector $b\in\bF_q^n$. For $h\in[3]$, define an isometry $\Enc^{(h)}:(\bC^q)^{\otimes k}\rightarrow(\bC^q)^{\otimes n}$ by $\Enc^{(h)}\ket{z}=\ket{\Enc^{(h)}_Z(z)}$. Then for every $\ket{\phi^{(1)}},\ket{\phi^{(2)}},\ket{\phi^{(3)}}\in(\bC^q)^{\otimes k}$,
  \begin{align*}
    \hspace{1em}&\hspace{-1em} \left(\bigotimes_{j\in[n]}CCZ_q^{b_j}\right)(\Enc^{(1)}\otimes\Enc^{(2)}\otimes\Enc^{(3)})\ket{\phi^{(1)},\phi^{(2)},\phi^{(3)}} \\
                &= (\Enc^{(1)}\otimes\Enc^{(2)}\otimes\Enc^{(3)})\left(\bigotimes_{j\in[k]}CCZ_q\right)\ket{\phi^{(1)},\phi^{(2)},\phi^{(3)}},
  \end{align*}
  where the $j$th $CCZ_q$ gate in the tensor products above acts on the $j$th qubit of the three length-$k$ or length-$n$ states given as the input.

  Similarly, if $Q=\CSS(Q_X,Q_Z;\Enc_Z)$ is an $[[n,k]]_q$ CSS code that supports a transversal $U_q$ gate with coefficients vector $b\in\bF_q^n$, then defining the isometry $\Enc:(\bC^q)^{\otimes k}\rightarrow(\bC^q)^{\otimes n}$ by $\Enc\ket{z}=\ket{\Enc_Z(z)}$, it follows that for every $\ket{\phi}\in(\bC^q)^{\otimes k}$,
  \begin{equation*}
    \left(\bigotimes_{j\in[n]}U_q^{b_j}\right)\Enc\ket{\phi} = \Enc\left(\bigotimes_{j\in[k]}U_q\right)\ket{\phi}.
  \end{equation*}
\end{lemma}
\begin{proof}
  For every $z^{(1)},z^{(2)},z^{(3)}\in\bF_q^k$, we have
  \begin{align*}
    \hspace{1em}&\hspace{-1em}\left(\bigotimes_{j\in[n]}CCZ_q^{b_j}\right)(\Enc^{(1)}\otimes\Enc^{(2)}\otimes\Enc^{(3)})\ket{z^{(1)},z^{(2)},z^{(3)}} \\
    &= \frac{1}{|Q_X^\perp|^{3/2}}\sum_{({z^{(h)}}'\in\Enc^{(h)}_Z(z))_{h\in[3]}}e^{2\pi i\tr_{\bF_q/\bF_p}(\sum_{j\in[n]}b_j\cdot {z^{(1)}_j}'{z^{(2)}_j}'{z^{(3)}_j}')/p}\ket{{z^{(1)}}',{z^{(2)}}',{z^{(3)}}'} \\
    &= \frac{1}{|Q_X^\perp|^{3/2}}\sum_{({z^{(h)}}'\in\Enc^{(h)}_Z(z))_{h\in[3]}}e^{2\pi i\tr_{\bF_q/\bF_p}(\sum_{j\in[k]}z^{(1)}_jz^{(2)}_jz^{(3)}_j)/p}\ket{{z^{(1)}}',{z^{(2)}}',{z^{(3)}}'} \\
    &= e^{2\pi i\tr_{\bF_q/\bF_p}(\sum_{j\in[k]}z^{(1)}_jz^{(2)}_jz^{(3)}_j)/p}(\Enc^{(1)}\otimes\Enc^{(2)}\otimes\Enc^{(3)})\ket{z^{(1)},z^{(2)},z^{(3)}} \\
    &= (\Enc^{(1)}\otimes\Enc^{(2)}\otimes\Enc^{(3)})\bigotimes_{j\in[k]}\left(e^{2\pi i\tr_{\bF_q/\bF_p}(z^{(1)}_jz^{(2)}_jz^{(3)}_j)/p}\ket{z^{(1)}_j,z^{(2)}_j,z^{(3)}_j}\right) \\
    &= (\Enc^{(1)}\otimes\Enc^{(2)}\otimes\Enc^{(3)})\left(\bigotimes_{j\in[k]}CCZ_q\right)\ket{z^{(1)},z^{(2)},z^{(3)}},
  \end{align*}
  where the second equality above holds by~(\ref{eq:supporttransCCZ}). Thus the desired result for the the $CCZ_q$ gate holds; the proof for the $U_q$ result is analogous.
\end{proof}

\section{Codes over Large Alphabets Supporting Transversal Gates}
\label{sec:latrans}
Our goal is to construct length-$n$ codes with distance and dimension close to $\Theta(n)$, which support a transversal $CCZ_q$ gate with an arbitrary constant alphabet size $q$. \cite{krishna_towards_2019} presented a construction of quantum punctured Reed-Solomon codes with length, distance, and dimension all $\Theta(n)$, but over a linear-sized alphabet $q=\Theta(n)$. In this section, we present an abstracted version of the construction of \cite{krishna_towards_2019}. We then show how to instantiate this construction with punctured algebraic geometry (AG) codes, which have the advantage over Reed-Solomon codes of having an alphabet size that is constant (but still slightly constrained). We will ultimately also present a way of further reducing the alphabet size to an arbitrary prime power.

\begin{theorem}
  \label{thm:generaltrans}
  Let $C$ be a $[n,\ell,d]_q$ classical code such that $C^{*3}:=\spn\{C*C*C\}$ has distance $d'$ and $C^\perp$ has distance $d^\perp$. Then for every $k<\min\{\ell,d,d',d^\perp\}$, there exists an $[[n-k,\; k,\; \geq\min\{d,d^\perp\}-k]]_q$ quantum CSS code that supports a transversal $CCZ_q$ gate.
\end{theorem}
\begin{proof}
  We define the desired CSS code $Q=\CSS(Q_X,Q_Z;\Enc_Z)$ as follows. Fix an arbitrary subset $A\subseteq[n]$ of size $|A|=k$ such that $C|_A=\bF_q^A$; such a set $A$ must exist by the assumption that $k<\ell$. Let $A^c=[n]\setminus A$, and let
  \begin{align*}
    Q_X &= C^\perp|_{A^c} \\
    Q_Z &= C|_{A^c}.
  \end{align*}
  Then
  \begin{equation*}
    Q_X^\perp = (C\cap(\{0\}^A\times\bF_q^{A^c}))|_{A^c} \subseteq Q_Z
  \end{equation*}
  because
  \begin{align*}
    ((C\cap(\{0\}^A\times\bF_q^{A^c}))|_{A^c})^\perp
    &= (C\cap(\{0\}^A\times\bF_q^{A^c}))^\perp|_{A^c} \\
    &= (C^\perp+(\bF_q^A\times\{0\}^{A^c}))|_{A^c} \\
    &= C^\perp|_{A^c} \\
    &= Q_X.
  \end{align*}
  
  Define $\Enc_Z:\bF_q^A\rightarrow\bF_q^{A^c}/Q_X^\perp$ by
  \begin{equation}
    \label{eq:genenc}
    \Enc_Z(z) = \{c|_{A^c}:c\in C,\;c|_A=z\}.
  \end{equation}

  To begin, we show that this definition indeed gives a valid isomorphism $\Enc_Z:\bF_q^A\xrightarrow{\sim}Q_Z/Q_X^\perp$. First, by the assumption that $C|_A=\bF_q^A$, we have $\Enc_Z(z)\neq\emptyset$. Now if $z'\in\Enc_Z(z)$, then by definition $z'+Q_X^\perp\subseteq\Enc_Z(z)$, as if $c\in C$ has $c|_A=z$, then $(c+c')|_A=z$ for every $c'\in C\cap(\{0\}^A\times\bF_q^{A^c})$. Conversely, if $z',z''\in\Enc_Z(z)$, then there exist $c',c''\in C$ with $c'|_A=c''|_A=z$ and $c'|_{A^c}=z'$, $c''|_{A^c}=z''$, so $(c'-c'')|_{A^c}\in Q_X^\perp$. Thus we must have $\Enc_Z(z)=z'+Q_X^\perp$. Therefore~(\ref{eq:genenc}) indeed gives a valid map $\Enc_Z:\bF_q^A\rightarrow\bF_q^{A^c}/Q_X^\perp$, which is also by definition linear. By definition every coset in $\im(\Enc_Z)$ lies inside $Q_Z$. Because $C$ has distance $d>k$, for every $z\neq 0$ we have $\Enc_Z(z)\neq 0$, so $\Enc_Z$ is injective. Furthermore, $\Enc_Z$ is surjective because $C|_A=\bF_q^A$, so $\Enc_Z:\bF_q^A\rightarrow Q_Z/Q_X^\perp$ is indeed an isomorphism. Therefore we have a well-defined CSS code $Q=\CSS(Q_X,Q_Z;\Enc_Z)$.
  
  By definition, $Q$ has length $|A^c|=n-k$ and dimension $|A|=k$. Furthermore, $Q_X=C^\perp|_{A^c}$ has distance $\geq d^\perp-|A|=d^\perp-k$, while $Q_Z\subseteq C|_{A^c}$ has distance $\geq d-k$, so $Q$ has distance $\geq\min\{d,d^\perp\}-k$.

  It remains to be shown that $Q$ supports a transversal $CCZ_q$ gate. For this purpose, we define a linear map $b:\bF_q^{A^c}\rightarrow\bF_q$ computed as follows. First, given a vector $z'\in C^{*3}|_{A^c}\subseteq\bF_q^{A^c}$, then by the assumption that $C^{*3}$ has distance $d'>k=|A|$, there exists a unique $c\in C^{*3}$ with $z'=c|_{A^c}$. Then let $b(z')=\sum_{j\in A}c_j$. Thus we have defined a linear map $b|_{C^{*3}}:C^{*3}\rightarrow\bF_q$; as $C^{*3}\subseteq\bF_q^{A^c}$, we may extend this map (in an arbitrary way) to a linear map $b:\bF_q^{A^c}\rightarrow\bF_q$. Then viewing $b$ as a vector in $\bF_q^{A^c}$ where $b(x)=\sum_{j\in A^c}b_jx_j$, then $b$ is our desired coefficients vector (see Definition~\ref{def:supporttrans}).

  Now fix $z^{(1)},z^{(2)},z^{(3)}\in\bF_q^A=\bF_q^k$ and consider any choice of representatives ${z^{(h)}}'\in\Enc_Z(z^{(h)})$ for $h\in[3]$. The assumption that $k<d$ implies that there exists a unique $c^{(h)}\in C$ satisfying ${z^{(h)}}'=c^{(h)}|_{A^c}$. Then $z^{(h)}=c^{(h)}|_A$, and the assumption that $k<d'$ ensures that $c^{(1)}*c^{(2)}*c^{(3)}$ is the unique element of $C^{*3}$ whose restriction to $A^c$ equals ${z^{(1)}}'{z^{(2)}}'{z^{(3)}}'$. Therefore by construction we have
  \begin{equation*}
    \sum_{j\in A^c}b_j\cdot{z^{(1)}_j}'{z^{(2)}_j}'{z^{(3)}_j}' = \sum_{j\in A}c^{(1)}_jc^{(2)}_jc^{(3)}_j = \sum_{j\in A}z^{(1)}_jz^{(2)}_jz^{(3)}_j,
  \end{equation*}
  which is precisely the desired equality~(\ref{eq:supporttransCCZ}).
\end{proof}

Thus if we have a family of constant-rate classical codes $C$ such that $C$, $C^{*3}$, and $C^\perp$ all have constant relative distance, then setting $k=\kappa n$ for a sufficiently small constant $0<\kappa<1$ in Theorem~\ref{thm:generaltrans}, we obtain an asymptotically good family of quantum CSS codes supporting a transversal $CCZ$ gate. \cite{krishna_towards_2019} instantiated this construction by letting $C=\evl{\bF_q}(\bF_q[X]^{<\ell})$ be a Reed-Solomon code (see Definition~\ref{def:polynomials}), which immediately yields the following:

\begin{corollary}[\cite{krishna_towards_2019}]
  \label{cor:ktconstruct}
  For every prime power $q$ and every $k,\ell\in\bN$ such that $k<\ell\leq q/2$ and $3(\ell-1)<n:=q-k$, there exists an $[[n,\; k,\; d=\ell+1-k]]_q$ CSS code that supports a transversal $CCZ_q$ gate.
\end{corollary}

We now show how to reduce the alphabet size by instead letting $C$ be an AG code. We first state a known result providing the desired AG code constructions. The reader is referred to \cite{hoholdt_algebraic_1998} for background on these codes.

\begin{theorem}[Well known; see e.g.~\cite{hoholdt_algebraic_1998}]
  \label{thm:agcodes}
  For every square prime power $q$ and every $L\in\bR$ satisfying
  \begin{equation}
    \label{eq:Lconstraint}
    \frac{1}{\sqrt{q}-1} < L < \frac13-\frac{1}{\sqrt{q}-1},
  \end{equation}
  there exists an infinite family $\cC$ of classical codes over the alphabet $\bF_q$ such that every $C\in\cC$ satisfies the following: if $C$ is an $[n,\ell,d]_q$ code for which $C^{*3}$ has distance $d'$ and $C^\perp$ has distance $d^\perp$, then
  \begin{align}
    \label{eq:agbounds}
    \begin{split}
      \ell/n &\geq L-o(1) \\
      d/n &\geq 1-L-\frac{1}{\sqrt{q}-1}-o(1) \\
      d'/n &\geq 1-3\left(L+\frac{1}{\sqrt{q}-1}\right)-o(1) \\
      d^\perp/n &\geq L-\frac{1}{\sqrt{q}-1}-o(1),
    \end{split}
  \end{align}
  where $o(1)$ denotes some function of $n$ that approaches $0$ as $n\rightarrow\infty$.
\end{theorem}
\begin{proof}
  The codes $C$ we choose are simply the codes $C(D,G)$ described in Theorem~2.65 of \cite{hoholdt_algebraic_1998}, instantiated with the asymptotically good curves described in Theorem~2.80 of \cite{hoholdt_algebraic_1998}. Specifically, for an infinite sequence of pairs $(n,g)$ with $g\leq n/(\sqrt{q}-1)+o(n)$, these curves have genus $g$ and have $n$ $\bF_q$-rational points. So we then instantiate $C=C(D,G)$ in Theorem~2.65 of \cite{hoholdt_algebraic_1998} with these curves, and by letting $\ell=\lceil Ln\rceil$ and choosing $G$ with $\deg(G)=\ell+g-1$. By~(\ref{eq:Lconstraint}), for every $n$ sufficiently large these choices satisfy the constraint $2g-2<\deg(G)<n$ on p.27 of \cite{hoholdt_algebraic_1998}, so their Theorem~2.65 implies that $C$ has dimension $\ell$ and distance $d\geq n-\deg(G)=n-\ell-g+1$. These inequalities imply the first two bounds in~(\ref{eq:agbounds}). By definition $C^{*3}\subseteq C(D,3G)$, and~(\ref{eq:Lconstraint}) also ensures that $2g-2<\deg(3G)=3\deg(G)<n$, so Theorem~2.65 of \cite{hoholdt_algebraic_1998} implies that $C^{*3}$ has distance $d'\geq n-3\deg(G)=n-3(\ell+g-1)$, which yields the third bound in~(\ref{eq:agbounds}). Meanwhile, Theorem~2.69 and Theorem~2.71 of \cite{hoholdt_algebraic_1998} imply that $C^\perp$ has distance $d^\perp\geq\deg(G)-2g+2=\ell-g+1$, yielding the fourth bound in~(\ref{eq:agbounds}).
\end{proof}

Applying the codes in Theorem~\ref{thm:agcodes} to Theorem~\ref{thm:generaltrans} immediately yields the following corollary. Note that for simplicity we do not try to optimize constants in the statement below.

\begin{corollary}
  \label{cor:transAG}
  For every square prime power $q\geq 64$, there exists an infinite subset $\cN\subseteq\bN$ such that for every $n\in\cN$, there is a $[[N\leq n,\; K\geq n/100,\; D\geq n/100]]_q$ quantum CSS code that supports a transversal $CCZ_q$ gate.
\end{corollary}
\begin{proof}
  Let $\cC$ be the family of classical codes over $\bF_q$ in Theorem~\ref{thm:agcodes} with $L=1/6$. Note that by assumption $\sqrt{q}-1\geq 7$, so~(\ref{eq:Lconstraint}) in Theorem~\ref{thm:agcodes} indeed holds. The desired quantum codes are then given by the quantum codes of dimension $K=\lceil n/100\rceil$ in Theorem~\ref{thm:generaltrans} with $C$ chosen from the family $\cC$. By Theorem~\ref{thm:generaltrans}, we only need to show that for every sufficiently large block length $n$ of $C\in\cC$, it holds that
  \begin{equation}
    \label{eq:checkcond1}
    K<\min\{\ell,d,d',d^\perp\}
  \end{equation}
  and that
  \begin{equation}
    \label{eq:checkcond2}
    \min\{d,d^\perp\}-K\geq n/100.
  \end{equation}
  The inequality~(\ref{eq:checkcond1}) holds because $\sqrt{q}-1\geq 7$, so by Theorem~\ref{thm:agcodes} we have $\min\{\ell,d,d',d^\perp\}/n\geq 1/42-o(1)$, which is greater than $K=\lceil n/100\rceil$ for every sufficiently large $n$. Meanwhile, the inequality~(\ref{eq:checkcond2}) holds because again by Theorem~\ref{thm:agcodes} we have $\min\{d,d^\perp\}/n\geq 1/42-o(1)$, which is greater than $K+n/100<n/45$ for every sufficiently large $n$.
\end{proof}

\section{Multiplication-Friendly Codes}
\label{sec:multfriend}
In this section, we describe multiplication-friendly codes, which have previously been used classically (see e.g.~\cite{guruswami_efficiently_2017}). We introduce a quantum analogue of such codes, for which we present multiple constructions. Specifically, we show that classical multiplication-friendly codes can be used to obtain quantum ones with poor distance, and we also present a more inherently quantum construction with better distance but with more restrictions on other parameters.

\begin{definition}
  \label{def:multfriendclass}
  Let $q$ be a prime power and let $m,n,k\in\bN$. For $h\in[m]$ let $C^{(h)}$ be a $[n,k]_q$ classical code with encoding function $\Enc^{(h)}:\bF_{q^k}\rightarrow\bF_q^n$. We say that the collection $(C^{(h)};\Enc^{(h)})_{h\in[m]}$ is \textbf{$m$-multiplication-friendly} if there exists an $\bF_q$-linear function $\Dec:\bF_q^n\rightarrow\bF_{q^k}$ such that for every $z_1,\dots,z_m\in\bF_{q^k}$, it holds that
  \begin{equation*}
    z_1z_2\cdots z_m = \Dec(\Enc^{(1)}(z_1)*\Enc^{(2)}(z_2)*\cdots*\Enc^{(m)}(z_m))
  \end{equation*}
  If $(C^{(h)};\Enc^{(h)})$ is the same for all $h$, we say that this code is $m$-multiplication-friendly.
\end{definition}

Below, we introduce a quantum analogue of Definition~\ref{def:multfriendclass}.

\begin{definition}
  \label{def:multfriend}
  Let $q$ be a prime power and let $m,n,k\in\bN$. For $h\in[m]$ let $Q^{(h)}=\CSS(Q^{(h)}_X,Q^{(h)}_Z;\Enc^{(h)}_Z)$ be a $[[n,k]]_q$ quantum CSS code with $Z$-encoding function $\Enc^{(h)}_Z:\bF_{q^k}\rightarrow\bF_q^n$. We say that the collection $(Q^{(h)})_{h\in[m]}$ is \textbf{$m$-multiplication-friendly} if there exists a $\bF_q$-linear function $\Dec_Z:\bF_q^n\rightarrow\bF_{q^k}$ such that for every $z_1,\dots,z_m\in\bF_{q^k}$, it holds for every $z_1'\in\Enc^{(1)}_Z(z_1),\dots,z_m'\in\Enc^{(m)}_Z(z_m)$ that
  \begin{equation}
    \label{eq:mfdec}
    z_1z_2\cdots z_m = \Dec(z_1'*z_2'*\cdots*z_m')
  \end{equation}
  If $Q^{(h)}$ is the same for all $h$, we say that this code is $m$-multiplication-friendly.
\end{definition}

The following lemma shows that classical multiplication-friendly codes can trivially be interpreted as quantum ones, albeit with no nontrivial distance bound.

\begin{lemma}
  \label{lem:ctoqMF}
  Let $(C^{(h)};\Enc^{(h)}:\bF_{q^k}\rightarrow\bF_q^n)_{h\in[m]}$ be an $m$-multiplication-friendly collection of $[n,k]_q$ classical codes. Then $(Q^{(h)}:=\CSS(\bF_q^n,C^{(h)};\Enc^{(h)}_Z=\Enc^{(h)}))_{h\in[m]}$ is an $m$-multiplication-friendly collection of $[[n,k]]_q$ quantum CSS codes.
\end{lemma}
\begin{proof}
  The result follows directly by Definition~\ref{def:multfriendclass} and Definition~\ref{def:multfriend}.
\end{proof}

We now describe two instantiations of classical multiplication-friendly codes. The first is simply the construction of \cite{guruswami_efficiently_2017}, but we include a proof for completeness.

\begin{lemma}[Follows from \cite{guruswami_efficiently_2017}]
  \label{lem:mfclass}
  For every prime power $q$ and every $n,k,m\in\bN$ such that $m(k-1)<n\leq q$, there exists an $m$-multiplication-friendly $[n,k]_q$ classical code.
\end{lemma}
\begin{proof}
  We define the desired code $C$ as the image of the encoding function $\Enc:\bF_{q^k}\rightarrow\bF_q^n$ described below. Fix an irreducible degree-$k$ polynomial $\gamma(X)\in\bF_q[X]$, so that $\bF_{q^k}=\bF_q[X]/(\gamma(X))$, and fix some subset $A\subseteq\bF_q$ of size $|A|=n$ (which is possible by the assumption that $n\leq q$). Then for $f(X)\in\bF_q[X]^{<k}\cong\bF_q[X]/(\gamma(X))=\bF_{q^k}$, define $\Enc(f)=\evl_A(f)$.

  Given $z_1,\dots,z_m\in\bF_q[X]^{<k}\cong\bF_q[X]/(\gamma(X))=\bF_{q^k}$, then letting $z_h'=\Enc(z_h)=\evl_A(z_h)$, we see that the product $z_1z_2\cdots z_m$ over $\bF_{q^k}$ corresponds to the quotient of the product $z_1z_2\cdots z_m$ over $\bF_q[X]$ modulo $\gamma(X)$. But this latter product (before quotienting by $\gamma(X)$) is a polynomial of degree $\leq m(k-1)$, and is therefore the unique such polynomial that agrees with $z_1'*z_2'*\cdots*z_m'$ at all points in $A$, as $|A|=n>m(k-1)$ by assumption.

  Thus let $\Dec:\bF_q^A\rightarrow\bF_{q^k}$ be the function defined as follows: given an input $z\in\bF_q^A$, let $f(X)\in\bF_q[X]$ be the unique polynomial of degree $<n$ that agrees with $z$ on every point in $A$, and define $\Dec(z)=f(X)/(\gamma(X))$, where we again use the equivalence $\bF_q[X]/(\gamma(X))=\bF_{q^k}$. By construction $\Dec$ is $\bF_q$-linear, and by the reasoning above we always have $z_1z_2\cdots z_m=\Dec_Z(z_1'*z_2'*\cdots*z_m')$, so $C=\im(\Enc)$ is $m$-multiplication-friendly, as desired.
\end{proof}

A disadvantage of the construction in Lemma~\ref{lem:mfclass} is that if $q=2$ and $m>1$, then we cannot satisfy the inquality $m(k-1)<q$. Thus this construction will not give multiplication-friendly codes over a binary alphabet. The following construction remedies this issue by using Reed-Muller codes instead of Reed-Solomon codes.

\begin{lemma}
  \label{lem:mfRM}
  For every prime power $q$ and every $k,m\in\bN$, there exists a collection $(C^{(h)};\Enc^{(h)})_{h\in[m]}$ of $m$-multiplication-friendly $[n:=q^{m(k-1)},\;k]_q$ classical codes.
\end{lemma}
\begin{proof}
  If $k=1$ then we can simply set every $C^{(h)}$ to be a length-$n$ repetition code, so assume that $k\geq 2$.
  
  We define the desired codes $C^{(h)}$ as the image of the encoding functions $\Enc^{(h)}:\bF_{q^k}\rightarrow\bF_q^n$ described below. Fix an irreducible degree-$k$ polynomial $\gamma(X)\in\bF_q[X]$, so that $\bF_{q^k}=\bF_q[X]/(\gamma(X))$. Let
  \begin{equation*}
    R = \bF_q[(X_{h,j})_{h\in[m],j\in[k-1]}]
  \end{equation*}
  be a ring of $m(k-1)$-variate polynomials over $\bF_q$, and define an ideal $I\subseteq R$ by
  \begin{equation*}
    I=(\gamma(X_{1,1}),(X_{h,j}-X_{1,1}^{j})_{h\in[m],j\in[k-1]}).
  \end{equation*}
  Observe that letting $X=X_{1,1}$ then $R/I\cong\bF_q[X]/(\gamma(X))=\bF_{q^k}$, as in $R/I$ we have that each $X_{h,j}=X_{1,1}^{j}$ is simply a power of $X_{1,1}$.

  Now for $h\in[m]$ and $f(X)=\sum_{j=0}^{k-1}f_{j}X^{j}\in\bF_q[X]^{<k}\cong\bF_q[X]/(\gamma(X))=\bF_{q^k}$, we define a degree-1 polynomial $F^{(h)}(f)\in R$ by
  \begin{equation*}
    F^{(h)}(f) = f_0+\sum_{j\in[k-1]}f_{j}X_{h,j},
  \end{equation*}
  and we let
  \begin{equation*}
    \Enc^{(h)}(f) = \evl_{\bF_q^{m(k-1)}}(F^{(h)}(f)).
  \end{equation*}
  Observe that the image of $F^{(h)}(f)$ under the quotient $R\rightarrow R/I$ is by definition $f(X_{1,1})$.

  We now show that the code collection $(C^{(h)}=\im(\Enc^{(h)}))_{h\in[m]}$ is $m$-multiplication-friendly. For this purpose, consider some $z_1,\dots,z_m\in\bF_q[X]^{<k}\cong\bF_q[X]/(\gamma(X))=\bF_{q^k}$, and for $h\in[m]$ let $z_h'=\Enc^{(h)}(z_h)=\evl_{\bF_q^{m(k-1)}}(F^{(h)}(z_h))$. By definition, the product $z_1z_2\cdots z_m$ over $\bF_{q^k}$ corresponds to the quotient of the product $z_1z_2\cdots z_m$ over $\bF_q[X]$ modulo $\gamma(X)$, which (letting $X=X_{1,1}$) equals the product $F^{(1)}(z_1)F^{(2)}(z_2)\cdots F^{(m)}(z_m)$ modulo $I$. But this latter product (before quotienting by $\gamma(X)$) is a multilinear polynomial of total degree $\leq m$, and is therefore the unique such polynomial that agrees with $z_1'*z_2'*\cdots*z_m'$ at all points in $\bF_q^{m(k-1)}$; this uniqueness follows from the fact that every nonzero multilinear polynomial (in fact, every polynomial with individual degree $<q$ for each variable) over $\bF_q$ has a nonzero evaluation point.

  Thus let $\Dec:\bF_q^{m(k-1)}\rightarrow\bF_{q^k}$ be the function defined as follows: given an input $z\in\bF_q^{m(k-1)}$, let $f(X_{h,j})_{h\in[m],j\in[k-1]}\in R$ be the unique polynomial with individual degree $<q$ in each variable that agrees with $z$ on every point in $\bF_q^{m(k-1)}$, and define $\Dec(z)\in\bF_{q^k}$ to be the image of $f$ under the quotient $R\rightarrow R/I=\bF_q[X_{1,1}]/(\gamma(X_{1,1}))=\bF_{q^k}$. By construction $\Dec$ is $\bF_q$-linear, and by the reasoning above we always have $z_1z_2\cdots z_m=\Dec(z_1'*z_2'*\cdots*z_m')$. Thus $(C^{(h)};\Enc^{(h)})_{h\in[m]}$ is $m$-multiplication-friendly, as desired.
\end{proof}

\begin{remark}
  \label{remark:diffcodes}
  In Lemma~\ref{lem:mfclass} we are able to use the same code $(C^{(h)};\Enc^{(h)})=(C;\Enc)$ for every $h\in[m]$, whereas Lemma~\ref{lem:mfRM} uses a different code $(C^{(h)};\Enc^{(h)})$ for each $h\in[m]$. We argue below that this difference is 1.~unavoidable, yet 2.~somewhat superficial:
  \begin{enumerate}
  \item In general, Lemma~\ref{lem:mfRM} cannot be strengthened to guarantee that all $(C^{(h)};\Enc^{(h)})$ be equal. Specifically, assume for a contradiction that $(C^{(h)};\Enc^{(h)})=(C;\Enc)$ for every $h\in[m]$ in Lemma~\ref{lem:mfRM}. Then when $q=2$ and $m\geq 3$, it would follow that for every $z_1,z_2\in\bF_{2^k}$,
  \begin{equation*}
    \Enc(z_1)*\Enc(z_2)^{*m-1} = \Enc(z_1)*\Enc(z_2) = \Enc(z_1)^{*2}*\Enc(z_2)^{*m-2},
  \end{equation*}
  where the equalities above hold because $x^t=x$ over $\bF_2$ for every $t\in\bN$. Then applying $\Dec$ to both sides gives that $z_1z_2^{m-1}=z_1^2z_2^{m-2}$. However for every $k\geq 2$ we can choose $z_1,z_2\in\bF_{2^k}$ for which $z_1z_2^{m-1}\neq z_1^2z_2^{m-2}$, yielding the desired contradiction. Note that Lemma~\ref{lem:mfclass} requires $m(k-1)<q$, and hence is not subject to this counterexample.
\item While strictly speaking all $(C^{(h)};\Enc^{(h)})$ are distinct in Lemma~\ref{lem:mfRM}, they are all isomorphic. Specifically, for every $h,h'\in[m]$, by definition there is an isomorphism $(C^{(h)};\Enc^{(h)})\cong(C^{(h')};\Enc^{(h')})$ given by (for instance) permuting the indices in $[m]$ in the proof of Lemma~\ref{lem:mfRM} by a cyclic shift of length $h'-h$. It follows that the codes $(C^{(h)};\Enc^{(h)})$ and $(C^{(h')};\Enc^{(h')})$ are equivalent up to a permutation of the $n$ code components.
\end{enumerate}
\end{remark}

While the classical multiplication-friendly codes above provide quantum multiplication-friendly codes by Lemma~\ref{lem:ctoqMF}, the only bound we have on the distance $d$ of these quantum codes is the trivial bound $d\geq 1$. Such a trivial bound will be sufficient for concatenating with the AG-based codes from Corollary~\ref{cor:transAG}, as for this purpose we only need a constant-sized multiplication-friendly code. Nevertheless, in Appendix~\ref{sec:qmultfriend}, we show how to obtain quantum multiplication-friendly codes with distance growing linearly in the block length. In Appendix~\ref{sec:RSconcat} we provide a sample application of these codes, in which we obtain quantum codes with nearly linear dimension and distance that support transversal $CCZ$ gates, using only concatenation with Reed-Solomon-based codes.

\section{Alphabet Reduction Procedure for Transversal Gates}
\label{sec:alphred}
In this section, we show how to concatenate a code supporting a transversal $CCZ_{q^r}$ gate with a multiplication-friendly code to obtain a code supporting a transversal $CCZ_q$ gate over a subfield $\bF_q\subseteq\bF_{q^r}$. A similar result naturally holds for the $U$ gate, but due to Lemma~\ref{lem:CCZtoU} it will be sufficient for us to restrict attention to the $CCZ$ gate here.

\begin{proposition}
  \label{prop:alphred}
  Let $(\MF{Q}^{(h)}=\CSS(\MF{Q}^{(h)}_X,\MF{Q}^{(h)}_Z;\MF{\Enc}^{(h)}_Z:\bF_{q^{\MF{k}}}\rightarrow\bF_{q^{\MF{n}}}))_{h\in[4]}$ be a quadruple of $4$-multiplication-friendly $[[\MF{n},\MF{k},\MF{d}^{(h)}]]_q$ CSS codes, and let $(Q^{(h)}=\CSS(Q^{(h)}_X,Q^{(h)}_Z;\Enc^{(h)}_Z))_{h\in[3]}$ be a triple of $[[n,k,d^{(h)}]]_{q^{\MF{k}}}$ CSS codes supporting a transversal $CCZ_{q^{\MF{k}}}$ gate.

  Also fix some $r\in\bN$ with $3(r-1)<\MF{k}$ and $r\leq q$, fix a subset $A\subseteq\bF_q$ of size $|A|=r$, and fix an irreducible degree-$\MF{k}$ polynomial $\gamma(X)\in\bF_q[X]$, so that $\bF_{q^{\MF{k}}}=\bF_q[X]/(\gamma(X))\cong\bF_q[X]^{<\MF{k}}$.
  
  For $h\in[4]$ let $\MF{\Enc}^{(h)}_X:\bF_{q^{\MF{k}}}\rightarrow\bF_{q^{\MF{n}}}$ be an arbitrary $X$ encoding function that is compatible with $\MF{\Enc}^{(h)}_Z$ under the trace bilinear form, and for $h\in[3]$ let $\Enc^{(h)}_X$ be an arbitrary $X$ encoding function that is compatible with $\Enc^{(h)}_Z$ under the standard bilinear form (see Definition~\ref{def:concatenation}).\footnote{Such $X$ encoding functions always exist by Lemma~\ref{lem:findcompatible}.} For $h\in[3]$, define a restricted concatenated CSS code
  \begin{equation*}
    \tilde{Q}^{(h)} = \MF{Q}^{(h)}\diamond_r Q^{(h)} := (\MF{Q}^{(h)}\circ Q^{(h)})|_{(\bF_q^A)^k},
  \end{equation*}
  where for the concatenation we use the $X,Z$ encoding functions above, and for the restriction we view
  \begin{equation}
    \label{eq:FqAsubspace}
    (\bF_q^A)^k\cong(\bF_q[X]^{<r})^k\subseteq(\bF_q[X]^{<\MF{k}})^k\cong(\bF_q[X]/(\gamma(X)))^k=\bF_{q^{\MF{k}}}^k
  \end{equation}
  as a $\bF_q$-linear subspace of $\bF_{q^{\MF{k}}}^k$, such that first isomorphism above is given by the inverse of $\evl_A:\bF_q[X]^{<r}\xrightarrow{\sim}\bF_q^A$.

  Then
  \begin{equation*}
    (\tilde{Q}^{(h)})_{h\in[3]} = \CSS(\tilde{Q}^{(h)}_X,\tilde{Q}^{(h)}_Z;\;\widetilde{\Enc}^{(h)}_Z:\bF_q^A\rightarrow(\bF_q^{\MF{n}})^n)
  \end{equation*}
  is a triple of $[[\MF{n}n,\;rk,\;\geq\MF{d}^{(h)}d^{(h)}]]_q$ CSS codes that supports a transversal $CCZ_q$ gate.
\end{proposition}

As the statement of Proposition~\ref{prop:alphred} is slightly involved, it may be helpful for the reader to first consider the case where $r=1$ (which is always a valid choice), in which case $\bF_q^A=\bF_q$ and~(\ref{eq:FqAsubspace}) simplifies to the natural inclusion $\bF_q^k\subseteq\bF_{q^{\MF{k}}}^k$. Increasing $r$ to larger values then serves to increase the dimension of the final codes $\tilde{Q}^{(h)}$.

\begin{proof}[Proof of Proposition~\ref{prop:alphred}]
  By Definition~\ref{def:concatenation} and Definition~\ref{def:restriction}, each $\tilde{Q}^{(h)}$ is a $[[\MF{n}n,rk,\geq \MF{d}^{(h)}d^{(h)}]]_q$ code, so it suffices to show that $(\tilde{Q}^{(h)})_{h\in[3]}$ supports a transversal $CCZ_q$ gate. For this purpose, let $b\in\bF_{q^{\MF{k}}}^n$ be the coefficients vector for $(Q^{(h)})_{h\in[3]}$ given by Definition~\ref{def:supporttrans}. We will define a coefficients vector $\tilde{b}\in(\bF_q^{\MF{n}})^n$ as follows. Let $\MF{\Dec}_Z:\bF_q^{\MF{n}}\rightarrow\bF_{q^{\MF{k}}}$ be the decoding function for $(\MF{Q}^{(h)})_{h\in[4]}$ as in Definition~\ref{def:multfriend}. Define an $\bF_q$-linear function $\eta:\bF_{q^{\MF{k}}}\rightarrow\bF_q$ so that for every $f(X)\in\bF_q[X]^{<\MF{k}}\cong\bF_q[X]/(\gamma(X))=\bF_{q^{\MF{k}}}$, we let $\eta(f)=\sum_{j\in A}f(j)$.
  
  Then for every $j\in[n]$, choose some $\MF{b}_j\in\MF{\Enc}^{(4)}_Z(b_j)$, and define an $\bF_q$-linear function $\tilde{b}_j:\bF_q^{\MF{n}}\rightarrow\bF_q$ by $\tilde{b}_j(z)=\eta(\MF{\Dec}_Z(\MF{b}_j*z))$. We can naturally view $\tilde{b}_j$ as a vector in $\bF_q^{\MF{n}}$, where $\tilde{b}_j(z)=\tilde{b}_j\cdot z$. Then we define our desired coefficients vector $\tilde{b}\in(\bF_q^{\MF{n}})^n$ by letting the $j$th component be $\tilde{b}_j$ for $j\in[n]$.

  For each $h\in[3]$, consider some $z^{(h)}\in(\bF_q^A)^k$, and define $f^{(h)}(X)=(\evl_A^{-1})^{\oplus k}(z^{(h)})\in(\bF_q[X]^{<r})^k$ by taking the preimage of each of the $k$ components of $z^{(h)}$ under the isomorphism $\evl:\bF_q[X]^{<r}\xrightarrow{\sim}\bF_q^A$. Therefore using the isomorphisms in~(\ref{eq:FqAsubspace}), we may view $z^{(h)}\cong f^{(h)}$ as an element of the $\bF_q$-linear subspace $(\bF_q^A)^k\cong(\bF_q[X]^{<r})^k$ of $\bF_{q^{\MF{k}}}^k$. Also consider some $\tilde{z}^{(h)}\in\widetilde{\Enc}^{(h)}_Z(z)$, so that by Definition~\ref{def:concatenation}, for some ${z^{(h)}}'\in\Enc^{(h)}_Z(z)$ we have $\tilde{z}^{(h)}\in(\MF{\Enc}^{(h)}_Z)^{\oplus n}({z^{(h)}}')$. Then
  \begin{align*}
    \sum_{j\in[n]}\sum_{\MF{j}\in[\MF{n}]}(\tilde{b}_j)_{\MF{j}} (\tilde{z}^{(1)}_j)_{\MF{j}}(\tilde{z}^{(2)}_j)_{\MF{j}}(\tilde{z}^{(3)}_j)_{\MF{j}}
    &= \sum_{j\in[n]}\tilde{b}_j\cdot(\tilde{z}^{(1)}_j*\tilde{z}^{(2)}_j*\tilde{z}^{(3)}_j) \\
    &= \sum_{j\in[n]}\eta(\MF{\Dec}_Z(\MF{b}_j*\tilde{z}^{(1)}_j*\tilde{z}^{(2)}_j*\tilde{z}^{(3)}_j)) \\
    &= \sum_{j\in[n]}\eta(b_j{z^{(1)}_j}'{z^{(2)}_j}'{z^{(3)}_j}') \\
    &= \eta\left(\sum_{j\in[k]}f^{(1)}_jf^{(2)}_jf^{(3)}_j\right) \\
    &= \sum_{j\in[k]}\sum_{\MF{j}\in A}(z^{(1)}_j)_{\MF{j}}(z^{(2)}_j)_{\MF{j}}(z^{(3)}_j)_{\MF{j}}.
  \end{align*}
  Note that the second equality above holds by the definition of $\tilde{b}_j$, and the third equality holds by the definition of $\MF{\Dec}_Z$ for a $4$-multiplication-friendly code. The fourth equality above holds because $\eta$ is $\bF_q$-linear, along with the assumption that $Q$ supports a transversal $CCZ_{q^{\MF{k}}}$ gate with coefficients vector $b$, as well as the assumption that $f^{(1)}f^{(2)}f^{(3)}$ has degree $\leq 3(r-1)<\MF{k}=\deg(\gamma)$ so that $f^{(1)}f^{(2)}f^{(3)}\mod{\gamma}=f^{(1)}f^{(2)}f^{(3)}$. The fifth equality holds by the definition of $\eta$ along with the fact that each $z^{(h)}_j=\evl_A(f^{(h)}_j)$ by definition. Thus $\tilde{Q}$ supports a transversal $CCZ_q$ gate, as desired.
\end{proof}

\section{Main Result via Alphabet Reduction}
\label{sec:agconcat}
In this section, we obtain infinite families of asymptotically good (i.e.~$[[n,\Theta(n),\Theta(n)]]_q$) CSS codes supporting transversal $CCZ_q$ and $U_q$ gates over arbitrary constant alphabet sizes $q$. The main idea is to begin with the quantum AG-based codes from Corollary~\ref{cor:transAG}, and then reduce the alphabet size to $q$ by applying Proposition~\ref{prop:alphred}, using the multiplication-friendly codes described in Section~\ref{sec:multfriend}.

\begin{theorem}
  \label{thm:constalph}
  For every fixed prime power $q$, there exists an infinite subset $\cN\subseteq\bN$ such that for every $n\in\cN$, there exists a triple $(Q^{(1)},Q^{(2)},Q^{(3)})$ of\footnote{As we treat $q$ as a fixed constant, the hidden constants in $\Theta$ here may depend on $q$.} $[[n,\Theta(n),\Theta(n)]]_q$ quantum CSS codes supporting a transversal $CCZ_q$ gate. Furthermore, if $q\geq 5$, then for every $n\in\cN$ the associated code triple has $Q^{(1)}=Q^{(2)}=Q^{(3)}$, and therefore these codes also support a transversal $U_q$ gate.
\end{theorem}
\begin{proof}
  If $q$ is a square and $q\geq 64$, then Corollary~\ref{cor:transAG} gives a family of codes with the desired properties. Thus assume that $q<64$ or that $q$ is not a square.
  
  To begin, we define a sequence of 4-multiplication-friendly quadruples $((\MF{Q}_t^{(h)})_{h\in[4]})_{t\in\bN}$ indexed by $t\in\bN$. Let $q_1=q$. Then given $q_t$ for $t\in\bN$, we define $(\MF{Q}_t^{(h)})_{h\in[4]}$ as follows:
  \begin{enumerate}
  \item\label{it:agqtl5} If $q_t<5$, let $(\MF{Q}_t^{(1)},\MF{Q}_t^{(2)},\MF{Q}_t^{(3)},\MF{Q}_t^{(4)})$ be the quadruple of $[[\MF{n}_t,\MF{k}_t]]_{q_t}$ codes from Lemma~\ref{lem:mfRM} with $m=4$ and $\MF{k}_t=4$. Thus each $\MF{Q}_t^{(h)}$ for $h\in[4]$ has parameters
    \begin{equation*}
      [[\MF{n}_t=q_t^{12},\; \MF{k}_t=4,\; \MF{d}_t\geq 1]]_{q_t}.
    \end{equation*}
  \item If $q_t\geq 5$, let $\MF{Q}_t^{(1)}=\MF{Q}_t^{(2)}=\MF{Q}_t^{(3)}=\MF{Q}_t^{(4)}$ be the $[[\MF{n}_t,\MF{k}_t]]_{q_t}$ code from Lemma~\ref{lem:mfclass} with $m=4$, $\MF{n}_t=q_t$, and $\MF{k}_t=2$ (which is valid as then $m(\MF{k}_t-1)=4<q_t$). Thus $\MF{Q}_t^{(h)}$ has parameters
    \begin{equation*}
      [[\MF{n}_t=q_t,\; \MF{k}_t=2,\; \MF{d}_t\geq 1]]_{q_t}.
    \end{equation*}
  \end{enumerate}
  To complete the inductive step defining $(\MF{Q}_t^{(h)})_{h\in[4]}$, we let $q_{t+1}=q_t^{\MF{k}_t}$. By construction we always have $\MF{k}_t\geq 2$, so $(q_t)_{t\in\bN}$ forms an infinite increasing sequence of powers of $q$.

  Now fix the least $T\in\bN$ such that $q_T\geq 64$. Let $\LAT{Q}$ be an arbitrary CSS code over $\bF_{q_T}$ from the family in Corollary~\ref{cor:transAG}, so that $\LAT{Q}$ has parameters $[[\LAT{n},\; \LAT{k}\geq\LAT{n}/100,\; \LAT{d}\geq\LAT{n}/100]]_{q_T}$ for some $\LAT{n}\in\bN$.

  Given such a choice of $\LAT{Q}$, we then define our desired codes $(Q^{(h)})_{h\in[3]}$ by
  \begin{equation*}
    Q^{(h)} := \MF{Q}_1^{(h)}\diamond_1\cdots\diamond_1\MF{Q}_{T-2}^{(h)}\diamond_1\MF{Q}_{T-1}^{(h)}\diamond_1\LAT{Q},
  \end{equation*}
  where $\diamond_1$ denotes the concatenation with restriction procedure given by setting $r=1$ in Proposition~\ref{prop:alphred}, and where we process the $\diamond_1$'s from right to left (so the innermost parentheses go around $\MF{Q}_{T-1}^{(h)}\diamond_1\LAT{Q}$). Note that when $q=q_1\geq 5$, we never need case~\ref{it:agqtl5} above in the definition of $(\MF{Q}_t^{(h)})_{h\in[4]}$, so we will have $Q^{(1)}=Q^{(2)}=Q^{(3)}$.

  By Proposition~\ref{prop:alphred}, the resulting triple $(Q^{(h)})_{h\in[3]}$ of codes supports a transversal $CCZ_q$ gate, and has parameters $[[n,k,d]]_q$ with
  \begin{align*}
    n &= \MF{n}_1\cdots\MF{n}_{T-1}\LAT{n} \leq q_1^{12}q_2\cdots q_{T-1}\LAT{n} \\
    k &= \LAT{k} \geq \LAT{n}/100 \\
    d &\geq \MF{d}_1\cdots\MF{d}_{T-1}\LAT{d} \geq \LAT{n}/100.
  \end{align*}
  Because $T$ is a constant depending only on $q$ as $\LAT{n}\rightarrow\infty$, it follows that $n,k,d=\Theta(\LAT{n})$. Thus these triples $(Q^{(h)})_{h\in[3]}$ (where we have one such triple for every $\LAT{Q}$ from Corollary~\ref{cor:transAG}) form an infinite family of CSS codes supporting a transversal $CCZ_q$ gate with paramters $[[n,\Theta(n),\Theta(n)]]_q$ for $n\rightarrow\infty$ as $\LAT{n}\rightarrow\infty$.
\end{proof}

We remark that in the proof of Theorem~\ref{thm:constalph} above, we made no attempt to optimize constant factors, but for practical applications such optimization could be important.

\section{Acknowledgments}
We thank Christopher A.~Pattison for numerous helpful discussions.

\bibliographystyle{alpha}
\bibliography{library}

\appendix

\section{Quantum Multiplication-Friendly Codes with Good Distance}
\label{sec:qmultfriend}
In this section, we present quantum multiplication-friendly codes with distance growing linearly with the block length. The construction uses similar methods as in Section~\ref{sec:latrans}. Recall that in contrast, the quantum multiplication-friendly codes presented in Section~\ref{sec:multfriend} were classical codes and thus only had distance $1$ as quantum codes. Note that to achieve linear distance, the construction in this section does impose more restrictions on the dimension and field size compared to Section~\ref{sec:multfriend}.

\begin{lemma}
  \label{lem:mfquant}
  For every prime power $q$ and every $k,r,\ell,m\in\bN$ such that $k\leq r\leq\ell\leq q/2$, $m(k-1)<r$, and $m(\ell-1)<n:=q-r$, there exists an $m$-multiplication-friendly $[[n,\; k,\; \geq\ell+1-r]]_q$ CSS code.
\end{lemma}
\begin{proof}
  Let $C=\evl_{\bF_q}(\bF_q[X]^{<\ell})\subseteq\bF_q^{\bF_q}$ be the classical $[q,\ell,q-\ell+1]_q$ Reed-Solomon code (see Definition~\ref{def:polynomials}). Fix some set $A\subseteq\bF_q$ of size $|A|=r$, and let $Q=\CSS(Q_X=C^\perp|_{A^c},\; Q_Z=C|_{A^c};\; \Enc_Z)$ be the CSS code with $Z$ encoding function $\Enc_Z:\bF_q[X]^A\rightarrow Q_Z/Q_X^\perp$ given by
  \begin{equation*}
    \Enc_Z(z) = \{c|_{A^c}:c\in C,\;c|_A=z\}.
  \end{equation*}
  As shown in the proof of Theorem~\ref{thm:generaltrans}, $Q$ is a well-defined $[[q-r,r,\geq\min\{d,d^\perp\}-r]]_q$ CSS code with $Q_X^\perp = (C\cap(\{0\}^A\times\bF_q^{A^c}))|_{A^c} \subseteq Q_Z$, where $d^\perp$ denotes the distance of $C^\perp$. Because $C^\perp=\evl(\bF_q[X]^{<q-\ell}$, we have $d^\perp=\ell+1$, so the assumption that $\ell\leq q/2$ ensures that $Q$ has distance at least $\min\{d,d^\perp\}-r=\ell+1-r$.

  Now fix an irreducible degree-$k$ polynomial $\gamma(X)\in\bF_q[X]$, so that $\bF_{q^k}=\bF_q[X]/(\gamma(X))\cong\bF_q[X]^{<k}$. We then define $Q'=Q|_{\bF_{q^k}}$ to be the restricted code (see Definition~\ref{def:restriction}) under the isomorphism
  \begin{equation*}
    \bF_{q^k}=\bF_q[X]/(\gamma(X))\cong\bF_q[X]^{<k}\cong\evl_A(\bF_q[X]^{<k}),
  \end{equation*}
  so that $Q'$ has $Z$ encoding function
  \begin{equation*}
    \Enc_Z' = \Enc_Z|_{\bF_{q^k}} : \bF_{q^k}=\bF_q[X]/(\gamma(X))\cong\bF_q[X]^{<k}\xrightarrow{\evl_{A}}\bF_q^A\xrightarrow{\Enc_Z}Q_Z/Q_X^\perp.
  \end{equation*}
  It follows by definition that $Q'$ is a $[[n,\; k,\; \geq\ell+1-r]]_q$ CSS code. We will now show that $Q'$ is $m$-multiplication-friendly.

  For this purpose, we first define the desired decoding function $\Dec':\bF_q^{A^c}\rightarrow\bF_{q^k}$ as follows. First, given an input $z'\in\bF_q^{A^c}$, let $f(X)\in\bF_q[X]^{<n}$ be the unique polynomial of degree $<n$ that agrees with $z'$ at every point in $A^c$. Then let $g(X)\in\bF_q[X]^{<r}$ be the unique polynomial of degree $<r$ that agrees with $f$ at every point in $A$, and define $\Dec'(z')=g(X)\mod\gamma(X)$, where we again use the equivalence $\bF_q[X]/(\gamma(X))=\bF_{q^k}$.

  To show that $\Dec'$ behaves as desired, we must show that~(\ref{eq:mfdec}) holds for $\Dec'$ for every choice of $z_1,\dots,z_m\in\bF_{q^k}=\bF_q[X]/(\gamma(X))\cong\bF_q[X]^{<k}$ and every $z_1'\in\Enc_Z'(z_1),\dots,z_m'\in\Enc_Z'(z_m)$. For this purpose, for each $h\in[m]$, let $f_h(X)\in\bF_q[X]^{<\ell}$ be the unique polynomial of degree $<\ell$ that agrees with $z_h'$ at all points in $A^c$. Then because $m(\ell-1)<n=|A^c|$, when $\Dec'$ is given input $z':=z_1'*z_2'*\cdots*z_m'$, the polynomial $f(X)\in\bF_q[X]^{<n}$ it interpolates from $z'$ must be the product $f(X)=f_1(X)f_2(X)\cdots f_m(X)$. It follows by the definition of $\Enc_Z'$ that $\evl_A(f_h)=\evl_A(z_h)$ for each $h\in[m]$, and therefore $\evl_A(f)=\evl_A(z_1z_2\cdots z_m)$, where we interpret each $z_h$ as a polynomial in $\bF_q[X]^{<k}$. Thus because $m(k-1)<r$, it follows that the polynomial $g(X)\in\bF_q[X]^{<r}$ that $\Dec'$ interpolates from $\evl_A(f)$ must equal $z_1z_2\cdots z_m$. Therefore $\Dec'(z')$ outputs $g(X)\mod \gamma(X)=z_1z_2\cdots z_m\mod\gamma(X)$, which by definition equals the product $z_1z_2\cdots z_m$ over $\bF_{q^k}=\bF_q[X]/(\gamma(X))$. Therefore~(\ref{eq:mfdec}) holds for $\Dec'$, as desired.
\end{proof}

\section{Nearly Good Construction via Iterative Alphabet Reduction}
\label{sec:RSconcat}
In this section, we obtain infinite families of CSS codes supporting transversal $CCZ_q$ and $U_q$ gates over arbitrary constant alphabet sizes $q$, such that the code dimension and distance grow very nearly linearly in the block length $n$. Formally, here we obtain dimension and distance at least $n/2^{O(\log^*n)}$, where $\log^*n$ is the very slow-growing function given by the number of times one must iteratively apply a logarithm starting from $n$ to obtain a value $\leq 1$.

Thus the construction in this section obtains slightly worse parameters than the asymptotically good codes of Theorem~\ref{thm:constalph}. However, here we simply use concatenation of codes based on Reed-Solomon (and for small alphabets, Reed-Muller) codes, which are more elementary compared to the AG codes used to prove Theorem~\ref{thm:constalph}.

Specifically, the main idea in the result below is to begin with the quantum Reed-Solomon codes of \cite{krishna_towards_2019} described in Corollary~\ref{cor:ktconstruct}, and then reduce the alphabet size to a constant via iterative applications of Proposition~\ref{prop:alphred}, using the multiplication-friendly codes described in Section~\ref{sec:multfriend} and Appendix~\ref{sec:qmultfriend}.

\begin{theorem}
  \label{thm:RSconstalph}
  For every fixed prime power $q$, there exists an infinite family $(Q_t^{(1)},Q_t^{(2)},Q_t^{(3)})_{t\in\bN}$ of quantum CSS codes supporting a transversal $CCZ_q$ gate with parameters\footnote{As we treat $q$ as a fixed constant, the hidden constants in $O,\Omega$ here may depend on $q$.}
  \begin{equation*}
    \left[\left[n_t,\; k_t\geq\Omega\left(\frac{n_t}{2^{O(\log^*n_t)}}\right),\; d_t\geq\Omega\left(\frac{n_t}{2^{O(\log^*n_t)}}\right)\right]\right]_q
  \end{equation*}
  such that $n_t\rightarrow\infty$ as $t\rightarrow\infty$. Furthermore, if $q\geq 5$, then we can ensure $Q_t^{(1)}=Q_t^{(2)}=Q_t^{(3)}$, and therefore these codes also support a transversal $U_q$ gate.
\end{theorem}
\begin{proof}
  To begin, we define an infinite sequence of 4-multiplication-friendly quadruples $((\MF{Q}_t^{(h)})_{h\in[4]})_{t\in\bN}$ indexed by $t\in\bN$. Let $q_1=q$. Then given $q_t$ for $t\in\bN$, we define $(\MF{Q}_t^{(h)})_{h\in[4]}$ as follows:
  \begin{enumerate}
  \item\label{it:qtl5} If $q_t<5$, let $(\MF{Q}_t^{(1)},\MF{Q}_t^{(2)},\MF{Q}_t^{(3)},\MF{Q}_t^{(4)})$ be the quadruple of $[[\MF{n}_t,\MF{k}_t]]_{q_t}$ codes from Lemma~\ref{lem:mfRM} with $m=4$ and $\MF{k}_t=3$. Thus each $\MF{Q}_t^{(h)}$ for $h\in[4]$ has parameters
    \begin{equation*}
      [[\MF{n}_t=q_t^8,\; \MF{k}_t=3,\; \MF{d}_t\geq 1]]_{q_t}.
    \end{equation*}
  \item If $5\leq q_t<200$, let $\MF{Q}_t^{(1)}=\MF{Q}_t^{(2)}=\MF{Q}_t^{(3)}=\MF{Q}_t^{(4)}$ be the $[[\MF{n}_t,\MF{k}_t]]_{q_t}$ code from Lemma~\ref{lem:mfclass} with $m=4$, $\MF{n}_t=q_t$, and $\MF{k}_t=\lfloor(q_t-1)/4\rfloor+1$. Thus $\MF{Q}_t^{(h)}$ has parameters
    \begin{equation*}
      [[\MF{n}_t=q_t,\; \MF{k}_t\geq 2,\; \MF{d}_t\geq 1]]_{q_t}.
    \end{equation*}
  \item\label{it:qtg200} If $q_t\geq 200$, let $\MF{Q}_t^{(1)}=\MF{Q}_t^{(2)}=\MF{Q}_t^{(3)}=\MF{Q}_t^{(4)}$ be the $[[\MF{n}_t,\MF{k}_t,\MF{d}_t]]_{q_t}$ code from Lemma~\ref{lem:mfquant} with $m=4$, $\MF{k}_t=\lfloor q_t/50\rfloor$, $\MF{r}_t=5\MF{k}_t$, $\MF{\ell}_t=10\MF{k}_t$, $\MF{n}_t=q_t-\MF{r}_t$, and $\MF{d}_t=\MF{\ell}_t+1-\MF{r}_t$, where $\MF{\ell}_t,\MF{r}_t$ correspond to the parameters in Lemma~\ref{lem:mfquant}. Thus $\MF{Q}_t^{(h)}$ has parameters
    \begin{equation*}
      [[\MF{n}_t\leq q_t,\; \MF{k}_t\in[q_t/100,q_t/50],\; \MF{d}_t\geq q_t/20]]_{q_t}.
    \end{equation*}
  \end{enumerate}
  To complete the inductive step defining $(\MF{Q}_t^{(h)})_{h\in[4]}$, we let $q_{t+1}=q_t^{\MF{k}_t}$. By construction we always have $\MF{k}_t\geq 2$, so $(q_t)_{t\in\bN}$ forms an infinite increasing sequence of powers of $q$.

  For every $t\geq 1$, we also define $r_t:=\lfloor(\MF{k}_t-1)/3\rfloor+1$. Then $3(r_t-1)<\MF{k}_t$, and also by the definition of the $\MF{k}_t$'s we have $r_t\leq q_t$.

  Now for every $t\geq 2$, by construction $q_t\geq 8$, so let $\LAT{k}_t=\lfloor q_t/8\rfloor$, $\LAT{\ell}_t=2\LAT{k}_t$, and $\LAT{n}_t=q_t-\LAT{k}_t$. Then let $\LAT{Q}_t$ be the $[[\LAT{n}_t,\LAT{k}_t,\LAT{d}_t]]_{q_t}$ code obtained by instantiating Corollary~\ref{cor:ktconstruct} with parameters $q_t,\LAT{k}_t,\LAT{\ell}_t$. Thus $\LAT{Q}_t$ has parameters
  \begin{equation*}
    [[\LAT{n}_t\leq q_t,\; \LAT{k}_t\in[q_t/16,q_t/8],\; \LAT{d}_t\geq q_t/8]]\\_{q_t}.
  \end{equation*}

  For $t\geq 2$, we then define our desired codes $(Q_t^{(h)})_{h\in[3]}$ by
  \begin{equation*}
    Q_t^{(h)} := \MF{Q}_1^{(h)}\diamond_{r_1}\cdots\diamond_{r_{t-2}}\MF{Q}_{t-2}^{(h)}\diamond_{r_{t-1}}\MF{Q}_{t-1}^{(h)}\diamond\LAT{Q}_t,
  \end{equation*}
  where $\diamond_r$ denotes the concatenation with restriction procedure given in Proposition~\ref{prop:alphred} with paramter $r$, and where we process the $\diamond_r$'s from right to left (so the innermost parentheses go around $\MF{Q}_{t-1}^{(h)}\diamond_{r_{t-1}}\LAT{Q}_t$). Note that when $q=q_1\geq 5$, we never need case~\ref{it:qtl5} above in the definition of $(\MF{Q}_u^{(h)})_{h\in[4]}$, so we will have $Q_t^{(1)}=Q_t^{(2)}=Q_t^{(3)}$.

  By Proposition~\ref{prop:alphred}, the resulting triple $(Q_t^{(h)})_{h\in[3]}$ of codes supports a transversal $CCZ_q$ gate, and has parameters $[[n_t,k_t,d_t]]_q$ with
  \begin{align*}
    n_t &= \MF{n}_1\cdots\MF{n}_{t-1}\LAT{n}_t \\
        &\leq q_1^8q_2\cdots q_t \\
    k_t &\geq r_1\cdots r_{t-1}\cdot\LAT{k}_t \\
        &\geq \frac{q_4}{400}\cdots\frac{q_{t-1}}{400}\cdot\frac{q_t}{16} \\
    d_t &\geq \MF{d}_1\cdots\MF{d}_{t-1}\LAT{d}_t \\
        &\geq \frac{q_4}{20}\cdots\frac{q_{t-1}}{20}\cdot\frac{q_t}{8},
  \end{align*}
  where the inequalities for $k_t,d_t$ above hold because by definition $q_2\geq 8$ and $q_3\geq 64$, so $q_4\geq 64^2$. Specifically, because $q_4\geq 200$, case~\ref{it:qtg200} above is used to define $(\MF{Q}_u^{(h)})_{h\in[4]}$ for every $u\geq 4$. Furthermore, for $u\geq 4$ it follows that $q_u\geq 4000$ and therefore that $r_u\geq\lfloor(q_u/100-1)/3\rfloor+1\geq q_u/400$.

  Now by definition for every $t>4$, we have $q_{t}\geq q_{t-1}^{q_{t-1}/100}\geq 2^{q_{t-1}/100}$, so
  $q_{t-1}\leq 100\log q_t$.
  Inductively applying this inequality implies that $t\leq O(\log^*(q_t))$, where recall that we assume $q=q_1$ is a fixed constant. Thus because $n_t\leq O(q_1q_2\cdots q_t)$ as shown above, it follows that
  \begin{align*}
    k_t &\geq \Omega\left(\frac{n_t}{400^t}\right) \geq \Omega\left(\frac{n_t}{2^{O(\log^*n_t)}}\right) \\
    d_t &\geq \Omega\left(\frac{n_t}{20^t}\right) \geq \Omega\left(\frac{n_t}{2^{O(\log^* n_t)}}\right).
  \end{align*}
\end{proof}

We remark that in the proof of Theorem~\ref{thm:RSconstalph} above, we made no attempt to optimize constant factors, but for practical applications such optimization could be important.

\end{document}